\documentclass[12pt]{article}
\usepackage{amsmath, amssymb,amsfonts,bbm,verbatim,multirow}
\usepackage[colorlinks,citecolor=blue,urlcolor=blue,breaklinks]{hyperref}
\usepackage{graphicx,xcolor}
\usepackage{amsthm}
\usepackage{fullpage}
\usepackage{natbib}
\usepackage{setspace}
\usepackage{enumerate}
\usepackage{array}
\usepackage[small]{caption}

\newtheorem{thm}{Theorem}

\newtheorem{cor}[thm]{Corollary}
\numberwithin{equation}{section}

\DeclareMathOperator*{\argmin}{argmin}

\begin{document}

\title{Comments on: High-dimensional simultaneous inference with the bootstrap}
\author{Richard A. Lockhart and Richard J. Samworth \\ Simon Fraser University and University of Cambridge \\ lockhart@sfu.ca, r.samworth@statslab.cam.ac.uk}

\maketitle





\section{Introduction}

We congratulate the authors on their stimulating contribution to the burgeoning high-dimensional inference literature.  The bootstrap offers such an attractive methodology in these settings, but it is well-known that its naive application in the context of shrinkage/superefficiency is fraught with danger \citep[e.g.][]{Samworth2003,ChatterjeeLahiri2011}.  The authors show how these perils can be elegantly sidestepped by working with de-biased, or de-sparsified, versions of estimators.

In this discussion, we consider alternative approaches to individual and simultaneous inference in high-dimensional linear models, and retain the notation of the paper.

\section{Why penalise coefficients of variables of interest?}

Suppose that for some, presumably small, set $G \subseteq \{1,\ldots,p\}$, we want a confidence set for $\beta_G^0$.  Much of the recent literature, including the paper under discussion, proceeds by constructing an initial estimator, such as the Lasso estimator $\hat{\beta}$, and then attempting to de-bias it.  Our starting point is the following provocative question: since we know in advance the set of variables we are interested in, why would we want to penalise these coefficients in the first place?  Of course, it is standard practice not to penalise the intercept term in high-dimensional linear models, to preserve location equivariance, but we now consider taking this one stage further.  More precisely, consider the linear model
\[
Y = \mathbf{X}\beta^0 + \epsilon,
\]
where the columns of $\mathbf{X}$ have Euclidean length $n^{1/2}$, where $\mathbf{X}_G^T\mathbf{X}_G$ is positive definite, and where, for simplicity, we assume that $\epsilon \sim N_n(0,\sigma^2I)$.  We further assume that the set $S := \{j : \beta_j^0 \neq 0\}$ of signal variables has cardinality $s$, and let $N := \{1,\ldots,p\} \setminus S$.  For $\lambda > 0$, let
\[
(\hat{\beta}_G,\hat{\beta}_{-G}) := \argmin_{(\beta_G,\beta_{-G}) \in \mathbb{R}^{|G|} \times \mathbb{R}^{p-|G|}} \frac{1}{n}\|Y - \mathbf{X}_G\beta_G - \mathbf{X}_{-G}\beta_{-G}\|_2^2 + \lambda \|\beta_{-G}\|_1,
\]
where we emphasise that $\|\beta_G\|_1$ is unpenalised.  For fixed $\beta_{-G} \in \mathbb{R}^{p-|G|}$, the solution in the first argument is given by ordinary least squares:
\[
\hat{\beta}_G(\beta_{-G}) := (\mathbf{X}_G^T\mathbf{X}_G)^{-1}\mathbf{X}_G^T(Y - \mathbf{X}_{-G}\beta_{-G}).
\] 
We therefore find that 
\begin{equation}
\label{Eq:Lasso}
\hat{\beta}_{-G} = \argmin_{\beta_{-G} \in \mathbb{R}^{p-|G|}} \frac{1}{n}\|(I-P_G)(Y - \mathbf{X}_{-G}\beta_{-G})\|_2^2 + \lambda \|\beta_{-G}\|_1,
\end{equation}
where $P_G := \mathbf{X}_G(\mathbf{X}_G^T\mathbf{X}_G)^{-1}\mathbf{X}_G^T$ denotes the matrix representing an orthogonal projection onto the column space of $\mathbf{X}_G$.  In other words, $\hat{\beta}_{-G}$ is simply the Lasso solution with response and design matrix pre-multiplied by $(I-P_G)$.  Moreover,
\[
\hat{\beta}_G = \hat{\beta}_G(\hat{\beta}_{-G}) = (\mathbf{X}_G^T\mathbf{X}_G)^{-1}\mathbf{X}_G^T(Y - \mathbf{X}_{-G}\hat{\beta}_{-G}).
\]
For our theoretical analysis of $\hat{\beta}_G$, we will require the following compatibility condition:
\begin{description}
\item[\textbf{(A1)}] There exists $\phi_0 > 0$ such that for all $b \in \mathbb{R}^{p-|G|}$ with $\|b_N\|_1 \leq 3\|b_S\|_1$, we have
\[
\|b_S\|_1^2 \leq \frac{s\|(I-P_G)\mathbf{X}_{-G}b\|_2^2}{n\phi_0^2}.
\]
\end{description}
The theorem below is only a small modification of existing results in the literature \citep[e.g.][]{BRT2009}, but for completeness we provide a proof in the Appendix.  
\begin{thm}
\label{Thm:BRT}
Assume \textbf{(A1)}, and let $\lambda := A\sigma\sqrt{\frac{\log p}{n}}$.  Then with probability at least $1 - p^{-(A^2/8-1)}$,
\[
\frac{1}{n}(\hat{\beta}_{-G} - \beta_{-G}^0)^T\mathbf{X}_{-G}^T(I-P_G)\mathbf{X}_{-G}(\hat{\beta}_{-G} - \beta_{-G}^0) + \frac{\lambda}{2}\|\hat{\beta}_{-G} - \beta_{-G}^0\|_1 \leq \frac{3A^2}{\phi_0^2}\frac{\sigma^2s\log p}{n}.
\]
\end{thm}
Theorem~\ref{Thm:BRT} allows us to show that if, in addition to \textbf{(A1)}, the columns of $\mathbf{X}_G$ and those of $\mathbf{X}_{-G}$ satisfy a strong lack of correlation condition, then $\hat{\beta}_G$ can be used for asymptotically valid inference for $\beta_G$.  To formalise this latter condition, it is convenient to let $\mathbf{\Theta}$ denote the $|G| \times (p-|G|)$ matrix $(\mathbf{X}_G^T\mathbf{X}_G)^{-1}\mathbf{X}_G^T \mathbf{X}_{-G}$.  
\begin{cor}
Consider an asymptotic framework in which $s=s_n \geq 1$ and $p=p_n \rightarrow \infty$ as $n \rightarrow \infty$, but $\sigma^2 > 0$ and $G$ are constant.  Assume~\textbf{(A1)} holds for sufficiently large $n$ (with $\phi_0$ not depending on $n$), and also that $\|\mathbf{\Theta}\|_\infty = o(s^{-1} \log^{-1/2} p)$.  If we choose $\lambda := A\sigma\sqrt{\frac{\log p}{n}}$ in the above procedure with constant $A > 2\sqrt{2}$, then
\[
n^{1/2}(\hat{\beta}_G - \beta_G^0) \stackrel{d}{\rightarrow} N_{|G|}\bigl(0,\sigma^2(\mathbf{X}_G^T\mathbf{X}_G)^{-1}\bigr). 
\]
\end{cor}
\begin{proof}
We can write
\[
n^{1/2}(\hat{\beta}_G - \beta_G^0) = n^{1/2}(\mathbf{X}_G^T\mathbf{X}_G)^{-1}\mathbf{X}_G^T\epsilon - \Delta,
\]
where $\Delta :=  n^{1/2}(\mathbf{X}_G^T\mathbf{X}_G)^{-1}\mathbf{X}_G^T \mathbf{X}_{-G}(\hat{\beta}_{-G} - \beta_{-G}^0)$.  Now
\[
n^{1/2}(\mathbf{X}_G^T\mathbf{X}_G)^{-1}\mathbf{X}_G^T\epsilon \sim N_{|G|}\bigl(0,\sigma^2(\mathbf{X}_G^T\mathbf{X}_G)^{-1}\bigr). 
\]
Moreover, from the proof of Theorem~\ref{Thm:BRT}, on $\Omega_0 := \bigl\{\|\mathbf{X}_{-G}^T(I-P_G)\epsilon\|_\infty/n \leq \lambda/2\}$, 
\[
\|\Delta\|_\infty \leq \|\mathbf{\Theta}\|_\infty n^{1/2}\|\hat{\beta}_{-G} - \beta_{-G}^0\|_1 \leq \frac{6A}{\phi_0^2} \|\mathbf{\Theta}\|_\infty s\log^{1/2} p \rightarrow 0. 
\]
Since $\mathbb{P}(\Omega_0) \rightarrow 1$, the conclusion follows.
\end{proof}
We remark that for $j \in G^c$, $\mathbf{\Theta}_j$ is the coefficient in the ordinary least squares regression of $X_j$ on $\mathbf{X}_G$.  Even though the condition on $\|\mathbf{\Theta}\|_\infty$ is strong, it may well be reasonable to suppose that, having pre-specified the index set $G$ of variables that we are interested in, we should avoid including in our model other variables that have significant correlation with $\mathbf{X}_G$.  

\section{More complicated settings}

Without this strong orthogonality condition we might instead consider adjusting $\hat{\beta}_G$ by debiasing or de-sparsifying $\hat{\beta}_{-G}$.  Following \citet{vandegeer2014} we suggest replacing $\hat{\beta}_{-G}$ by
$$
\hat{b}_{-G} = \hat{\beta}_{-G} + \frac{1}{n} M \mathbf{X}_{-G}^T(I-P_G)(Y - \mathbf{X}_{-G}\hat{\beta}_{-G})
$$
for some matrix $M \in \mathbb{R}^{(p-|G|)\times (p-|G|)}$.  This yields the debiased estimator
\begin{align*}
\hat{b}_G &= (\mathbf{X}_G^T\mathbf{X}_G)^{-1}\mathbf{X}_G^T(Y - \mathbf{X}_{-G}\hat{b}_{-G})
\\
& = \beta_{G}^0 +  (\mathbf{X}_G^T\mathbf{X}_G)^{-1}\mathbf{X}_G^T \epsilon - \frac{1}{n} \boldsymbol{\Theta}M \mathbf{X}_{-G}^T(I-P_G)\epsilon - R(\hat{\beta}_{-G}-\beta_{-G}^0),
\end{align*}
where $R$ is the $|G| \times (p-|G|)$ matrix given by
$$
R := \boldsymbol{\Theta} - \frac{1}{n} \boldsymbol{\Theta}M \mathbf{X}_{-G}^T(I-P_G)\mathbf{X}_{-G}.
$$
Under our Gaussian errors assumption, $(\mathbf{X}_G^T\mathbf{X}_G)^{-1}\mathbf{X}_G^T \epsilon$ and $n^{-1}\boldsymbol{\Theta} M \mathbf{X}_{-G}^T(I-P_G)\epsilon$ are independent centred Gaussian random vectors; thus if the remainder term $R(\hat{\beta}_{-G}-\beta_{-G}^0)$ is of smaller order, we see that our estimate $\hat{b}_G$ is approximately centred Gaussian.  The techniques of \citet{vandegeer2014} or  \citet{javanmard2014a} might then be used to give asymptotic justifications for Gaussian confidence sets and hypothesis tests concerning $\beta_G^0$.  But another very interesting direction would be to adapt the bootstrap approaches proposed in the current paper to the estimate $\hat{b}_G$.

As in \citet{vandegeer2014} we should choose $M$ depending on $\mathbf{X}$ to control 
$$
\delta:=\|R(\hat{\beta}_{-G}-\beta_{-G})\|_\infty\le \|R\|_\infty \|\hat{\beta}_{-G}-\beta_{-G}\|_1.
$$  
Note that we may write the matrix $R$ in terms of the sample covariance matrix of the covariates $\hat\Sigma :=\mathbf{X}^T\mathbf{X}/n$ (using obvious notation for the partitioning) as
$$
R = \hat\Sigma_{G,G}^{-1}\hat\Sigma_{G,-G} \bigl(I-M(\hat\Sigma_{-G,-G} -\hat\Sigma_{-G,G}\hat\Sigma_{G,G}^{-1}\hat\Sigma_{G,-G})\bigr).
$$
Of course, if $\hat{\Sigma}$ is invertible, then $(\hat\Sigma_{-G,-G} -\hat\Sigma_{-G,G}\hat\Sigma_{G,G}^{-1}\hat\Sigma_{G,-G})^{-1} = (\hat{\Sigma}^{-1})_{-G,-G}$, so $M$ can be thought of as an approximation to $(\hat{\Sigma}^{-1})_{-G,-G}$ (even though $\hat{\Sigma}$ is not invertible when $p > n$).  In general, we might use concentration inequalities for entries in $\hat\Sigma$ to control $\|R\|_\infty$; if we think of $|G|$ as small, then we only have $O(p)$ entries to control, rather than $O(p^2)$ as is more typical in these debiasing problems.  We hope to pursue these ideas elsewhere.

\section*{Acknowledgements}

The first author thanks St John's College, Cambridge for kind hospitality over the period where this research was carried out.  The second author is support by an Engineering and Physical Sciences Research Council Fellowship and a grant from the Leverhulme Trust.

\section*{Appendix}

\begin{proof}[Proof of Theorem~\ref{Thm:BRT}]
The KKT conditions for the problem~\eqref{Eq:Lasso} state that
\[
\frac{1}{n}\mathbf{X}_{-G}^T(I-P_G)(Y - \mathbf{X}_{-G}\hat{\beta}_{-G}) = \lambda \gamma,
\]
where $\|\gamma\|_\infty \leq 1$ and $\gamma_j = \mathrm{sgn}(\hat{\beta}_{-G,j})$ if $\hat{\beta}_{-G,j} \neq 0$.  Thus
\begin{align*}
\frac{1}{n}(\beta_{-G}^0 - \hat{\beta}_{-G})^T&\mathbf{X}_{-G}^T(I-P_G)\mathbf{X}_{-G}(\beta_{-G}^0 - \hat{\beta}_{-G}) \\
&= \lambda (\beta_{-G}^0 - \hat{\beta}_{-G})^T\gamma - \frac{1}{n}(\beta_{-G}^0 - \hat{\beta}_{-G})^T\mathbf{X}_{-G}^T(I-P_G)\epsilon \\
&= \lambda (\beta_{-G}^0)^T\gamma - \lambda\|\hat{\beta}_{-G}\|_1 - \frac{1}{n}(\beta_{-G}^0 - \hat{\beta}_{-G})^T\mathbf{X}_{-G}^T(I-P_G)\epsilon \\
&\leq \lambda\|\beta_{-G,S}^0\|_1 - \lambda\|\hat{\beta}_{-G}\|_1 + \|\hat{\beta}_{-G} - \beta_{-G}^0\|_1 \frac{1}{n}\|\mathbf{X}_{-G}^T(I-P_G)\epsilon\|_\infty.
\end{align*}
Let $\Omega_0 := \bigl\{\|\mathbf{X}_{-G}^T(I-P_G)\epsilon\|_\infty/n \leq \lambda/2\}$.  Then since $\mathbf{X}_{-G}^T(I-P_G)\epsilon \sim N_p(0,\sigma^2\mathbf{X}_{-G}^T(I-P_G)\mathbf{X}_{-G})$, and since the diagonal entries of $\mathbf{X}_{-G}^T(I-P_G)\mathbf{X}_{-G}$ are bounded above by $n$, we have $\mathbb{P}(\Omega_0^c) \leq p^{-(A^2/8-1)}$.  Moreover, on $\Omega_0$,
\begin{align*}
\frac{1}{n}&(\hat{\beta}_{-G} - \beta_{-G}^0)^T\mathbf{X}_{-G}^T(I-P_G)\mathbf{X}_{-G}(\hat{\beta}_{-G} - \beta_{-G}^0) + \frac{\lambda}{2}\|\hat{\beta}_{-G,N}\|_1 \\
&= \frac{1}{n}(\hat{\beta}_{-G} - \beta_{-G}^0)^T\mathbf{X}_{-G}^T(I-P_G)\mathbf{X}_{-G}(\hat{\beta}_{-G} - \beta_{-G}^0) + \lambda\|\hat{\beta}_{-G}\|_1 - \lambda\|\hat{\beta}_{-G,S}\|_1 - \frac{\lambda}{2}\|\hat{\beta}_{-G,N}\|_1 \\
&\leq \frac{\lambda}{2}\|\hat{\beta}_{-G} - \beta_{-G}^0\|_1 - \frac{\lambda}{2}\|\hat{\beta}_{-G,N}\|_1 + \lambda(\|\beta_{-G,S}^0\|_1 - \|\hat{\beta}_{-G,S}\|_1) \\
&\leq \frac{3\lambda}{2}\|\hat{\beta}_{-G,S} - \beta_{-G,S}^0\|_1.
\end{align*}
In particular, $\|\hat{\beta}_{-G,N} - \beta_{-G,N}^0\|_1 = \|\hat{\beta}_{-G,N}\|_1 \leq 3\|\hat{\beta}_{-G,S} - \beta_{-G,S}^0\|_1$, so from~\textbf{(A1)},
\begin{align*}
\frac{1}{n}\|(I-P_G)\mathbf{X}_{-G}(\hat{\beta}_{-G} - \beta_{-G}^0)\|_2^2 + \frac{\lambda}{2}\|\hat{\beta}_{-G,N}\|_1 &\leq \frac{3\lambda}{2}\|\hat{\beta}_{-G,S} - \beta_{-G,S}^0\|_1 \\
&\leq \frac{3\lambda}{2}\frac{s^{1/2}\|(I-P_G)\mathbf{X}_{-G}(\hat{\beta}_{-G} - \beta_{-G}^0)\|_2}{n^{1/2}\phi_0}.
\end{align*}
Thus
\[
\frac{1}{n^{1/2}}\|(I-P_G)\mathbf{X}_{-G}(\hat{\beta}_{-G} - \beta_{-G}^0)\|_2 \leq \frac{3\lambda s^{1/2}}{2\phi_0}.
\]
We conclude that 
\begin{align*}
\frac{1}{n}(\hat{\beta}_{-G} - \beta_{-G}^0)^T\mathbf{X}_{-G}^T&(I-P_G)\mathbf{X}_{-G}(\hat{\beta}_{-G} - \beta_{-G}^0) + \frac{\lambda}{2}\|\hat{\beta}_{-G} - \beta_{-G}^0\|_1 \leq 2\lambda\|\hat{\beta}_{-G,S} - \beta_{-G,S}^0\|_1 \\
&\leq \frac{2\lambda s^{1/2}\|(I-P_G)\mathbf{X}_{-G}(\hat{\beta}_{-G} - \beta_{-G}^0)\|_2}{n^{1/2}\phi_0} \\
&\leq \frac{3A^2}{\phi_0^2}\frac{\sigma^2s\log p}{n},
\end{align*}
as required.
\end{proof}


\begin{thebibliography}{99}
\bibitem[{Bickel, Ritov and Tsybakov(2009)}]{BRT2009}Bickel, P. J., Ritov, Y. and Tsybakov, A. B. (2009) Simultaneous analysis of Lasso and Dantzig selector.
\newblock \emph{Ann. Statist.}, \textbf{37}, 1705--1732.
\bibitem[{Chatterjee and Lahiri(2011)}]{ChatterjeeLahiri2011}Chatterjee, A. and Lahiri, S. N. (2011) Bootstrapping Lasso estimators.
\newblock \emph{J. Amer. Statist. Assoc.}, \textbf{106}, 608--625. 
\bibitem[{Javanmard and Montanari(2014)}]{javanmard2014a} Javanmard, A. and Montanari, A. (2014) Confidence Intervals and Hypothesis Testing for High-Dimensional Regression.
\newblock \emph{J.  Machine Learning Res.},
 \textbf{15}, 2869--2909.
\bibitem[{Samworth(2003)}]{Samworth2003}Samworth, R. (2003) A note on methods of restoring consistency to the bootstrap.
\newblock \emph{Biometrika}, \textbf{90}, 985--990.
\bibitem[{van de Geer et al.(2014)}]{vandegeer2014} van de Geer, S., B\"uhlmann, P.,  Ritov, Y. and Dezeure, R. (2014) On asymptotically optimal confidence regions and tests for high-dimensional models. 
\newblock \emph{Ann. Statist.}, \textbf{42}, 1166--1202.
\end{thebibliography}
\end{document}